\documentclass[11pt,a4paper]{article}
\usepackage[T1]{fontenc}
\usepackage[utf8]{inputenc}
\usepackage{lmodern,microtype}
\usepackage[margin=25mm,headheight=14pt]{geometry}
\usepackage{amsmath,amssymb,amsthm,mathtools}
\usepackage{graphicx,booktabs,enumitem,placeins}
\usepackage{fancyhdr}
\usepackage{float}
\usepackage[hidelinks]{hyperref}
\setlist[enumerate]{leftmargin=2em,itemsep=2pt,topsep=4pt}
\allowdisplaybreaks[1]
\newtheorem{theorem}{Theorem}
\newtheorem{lemma}{Lemma}

\newtheorem*{conjecture}{Conjecture 4}
\theoremstyle{remark}
\newcommand{\R}{\mathbb R}
\newcommand{\X}{\mathcal X}
\newcommand{\Kinf}{\mathcal K_\infty}
\newcommand{\KL}{\mathcal{KL}}
\newcommand{\dd}{\,\mathrm d}
\newcommand{\norm}[1]{\lVert#1\rVert}
\newcommand{\D}{D^+}

\newcommand\blfootnote[1]{%
  \begingroup
  \renewcommand\thefootnote{}\footnote{#1}%
  \addtocounter{footnote}{-1}%
  \endgroup
}

\title{\vspace{-1.2em}\textbf{A Candidate Counterexample to a Conjecture on ISS for Time-Delay Systems}}
\author{Vittorio De Iuliis\thanks{Department for the Promotion of Human Science and Quality of Life, San Raffaele University of Rome, Italy. \texttt{vittorio.deiuliis@uniroma5.it}.}\qquad
Pierdomenico Pepe\thanks{Department of Information Engineering, Computer Science and Mathematics, University of L'Aquila, Italy. \texttt{pierdomenico.pepe@univaq.it}.}}
\date{}
\begin{document}
\maketitle\thispagestyle{plain}
\blfootnote{\textit{Preprint submitted to arXiv on September 18, 2026.}}
\vspace{-10pt}
\begin{abstract}
We present a candidate counterexample to a conjecture stating that the existence of a Lyapunov-Krasovskii functional with a pointwise dissipation rate is sufficient for the input-to-state stability of time-delay systems. 
The counterexample has been derived through interactions with large language models.
\end{abstract}
\vspace{5pt}

\vspace{5pt}

\section{Introduction and the conjecture}
For retarded systems, a Lyapunov--Krasovskii functional (LKF) depends on a history segment, whereas a pointwise dissipation rate depends only on the current solution value.  Conjecture 4 in~\cite[Section~8.3, p.~297]{survey} states that pointwise dissipation guarantees ISS.

Here we provide a candidate counterexample to such Conjecture. In particular, we present a two-dimensional retarded functional differential equation, and a piecewise continuous bounded input signal, such that the solution is unbounded. Thus, the system is not ISS. On the other hand, we present a Lyapunov-Krasovskii functional which is Lipschitz on bounded sets, which exhibits a pointwise dissipation rate as in Conjecture 4 in~\cite[Section~8.3, p.~297]{survey}. 

Such counterexample was obtained through repeated interaction with the latest OpenAI's ChatGPT large language model (LLM) GPT-6 Astra. Initially, we attempted to work toward a proof of the conjecture. Having made no substantial progress, we tried asking for a counterexample to it. The model produced a counterexample based on the same underlying idea as the one presented here, but whose theoretical and numerical verification posed significant difficulties: the system dynamics were governed by exponential gain and damping terms, with state-dependent input activation times occurring over infinitesimal intervals (less than $10^{-58}$\textit{s}) that greatly complicated the numerical validation. Through repeated interactions, we succeeded in obtaining the counterexample reported here, which utilizes polynomial gain and damping terms in the system dynamics alongside input functions with state-independent activations.

We are intrigued and challenged by the depth of analysis demonstrated by frontier LLMs, which yielded this candidate counterexample, and by the broader consequences of these emerging capabilities for scientific discovery. Up to our best efforts in checking the technical validity of the arguments, the counterexample is true. The same conclusion has been confirmed by Anthropic's LLM Claude Fable 5.

\paragraph{Notation.}
Set $\X=C([-1,0];\R^2)$ and
$\norm{\phi}=\max_{-1\le\theta\le0}|\phi(\theta)|$, where $| \cdot |$ is the Euclidean norm in $\mathbb{R}^2$. For positive real $T$, $x:[-1,T)\to\R^2$, let $x_t(\theta)=x(t+\theta)$. We consider
\begin{equation} \label{eq:sistema}
 \dot x(t)=f(x_t,u(t)),\qquad x_0=\phi\in\X,
\end{equation}
with measurable, locally essentially bounded inputs $u : \mathbb{R}_{\geq 0} \to \mathbb{R}^2$, $f: \mathcal{X}\times \mathbb{R}^2 \to \mathbb{R}^2$ Lipschitz on bounded sets and $f(0,0)=0$. These hypotheses give existence, uniqueness, continuous dependence, and the boundedness implies continuation property~\cite[Theorem~2]{survey}. 

A continuous, strictly increasing map $\alpha:\R_{\ge0}\to\R_{\ge0}$ with $\alpha(0)=0$ belongs to $\mathcal K$; if it is also unbounded, it belongs to $\Kinf$. The class $\mathcal N$ consists of continuous nondecreasing maps vanishing at zero. A continuous function $\beta$: $\mathbb{R}_{\geq 0} \times \mathbb{R}_{\geq 0} \to \mathbb{R}_{\geq 0}$ belongs to $\KL$ if, for any $t \in \mathbb{R}_{\geq 0}$, $\beta(\cdot,t)$ is of class $\mathcal K$ and, for any $s \in \mathbb{R}_{\geq 0}$,  $\beta(s,\cdot)$ decreases to zero. An LKF is a functional $V:\X\to\R_{\ge0}$, Lipschitz on bounded sets, such that for some $\underline\alpha,\overline\alpha\in\Kinf$: 
\begin{equation}
 \underline\alpha(|\phi(0)|)\le V(\phi)\le\overline\alpha(\norm{\phi}), \quad \forall \phi \in \mathcal{X}.
 \label{eq:lkf}
\end{equation}

For $w\in\R^2$ and $0<h<1$, define
\begin{equation}
 \phi_{h,w}(\theta)=
 \begin{cases}
 \phi(\theta+h),&-1\le\theta<-h,\\
 \phi(0)+(\theta+h)w,&-h\le\theta\le0,
 \end{cases}
 \qquad
 \D V(\phi,w)=\limsup_{h\downarrow0}\frac{V(\phi_{h,w})-V(\phi)}h.
 \label{eq:driverdef}
\end{equation}

 System \eqref{eq:sistema} is ISS if there exist 
$\beta\in\KL$ and $\mu\in\mathcal N$ such that, for any initial state $\phi \in \mathcal{X}$ and any Lebesgue-measurable locally essentially bounded input $u : \mathbb{R}_{\geq 0} \to \mathbb{R}^2$, the solution of system \eqref{eq:sistema} exists for all $t \in \mathbb{R}_{\geq 0}$ and furthermore satisfies:
\begin{equation}
|x(t,\phi,u)|\le\beta(\norm{\phi},t)+\mu(\norm{u_{[0,t]}}_\infty).
 \label{eq:iss}
\end{equation}
\begin{conjecture}[{\cite[Section~8.3]{survey}}]
If an LKF satisfying~\eqref{eq:lkf} admits $\alpha\in\Kinf$ and $\gamma\in\mathcal N$ such that
\begin{equation}
 \D V(\phi,f(\phi,u))\le-\alpha(|\phi(0)|)+\gamma(|u|)
 \quad\text{for all }\phi\in\X,\ u\in\R^2,
 \label{eq:conjecture}
\end{equation}
then system \eqref{eq:sistema} is ISS.
\end{conjecture}

\section{Explicit functional, dynamics, and input}
Write $x= [ p \ \ q ]^\top $, $u= [ a \ \ b ]^\top $ and $\phi= [ \pi \ \ \zeta ]^\top $, with $p=\pi(0)$, $q=\zeta(0)$. For $\sigma\in\{+1,-1\}$, let $g_\sigma: \mathbb{R} \to \mathbb{R}$ be defined as:
\begin{equation} 
 g_\sigma(v)=\tfrac12\max\{\sigma v,0\}^2,
 \qquad
\end{equation}
so its derivative $g_\sigma': \mathbb{R} \to \mathbb{R}$ is given by
\begin{equation}
g_\sigma'(v)=\begin{cases}v,&\sigma v>0,\\0,&\sigma v\le0.\end{cases}
 \label{eq:g}
\end{equation}
Thus $g_\sigma\in C^1$, $g_\sigma'$ is globally $1$-Lipschitz, and
$\max_\sigma g_\sigma(v)=v^2/2$. Define
\begin{align}
 F_\sigma(\phi,\theta)
 &=g_\sigma(\pi(\theta))+\tfrac12q^2+g_\sigma(q)
 -g_\sigma(\zeta(\theta))-\int_\theta^0\pi(s)^2\dd s,
 \label{eq:F}\\
 V(\phi)&=\max_{\substack{\sigma\in\{+1,-1\}\\\theta\in[-1,0]}}
 F_\sigma(\phi,\theta),\qquad
 C(\phi)=\tfrac12(p^2+q^2),\qquad E(\phi)=V(\phi)-C(\phi).
 \label{eq:V}
\end{align}
The proposed time-invariant system is
\begin{equation}
 \boxed{\begin{aligned}
 \dot p(t)&=-p(t)-p(t)^5+16\big(1+V(x_t)\big)^3b(t)E(x_t),\\
 \dot q(t)&=-q(t)+a(t).
 \end{aligned}}
 \label{eq:system}
\end{equation}

Notice that the system is globally exponentially stable in the zero-input case.

For integers $n\ge0$, set
\begin{equation}
 \varepsilon_n=(-1)^n,\qquad
 \tau_n=\frac n4+\frac1{64}\sum_{j=1}^n\frac1j,\qquad
 \rho_n=\tau_n+\frac14,\qquad
 \delta_n=\frac1{64(n+1)}.
 \label{eq:times}
\end{equation}
The sum is zero for $n=0$, and $\tau_{n+1}=\rho_n+\delta_n$. The input is given by:
\begin{equation} 
 \boxed{u(t)= \begin{bmatrix} a(t) \\ b(t)\end{bmatrix} =  \begin{cases}  
 \begin{bmatrix}
 \varepsilon_n[28+32(t-\tau_n)] \\ 0 \end{bmatrix},  &\tau_n\le t<\rho_n,\\ \\ \begin{bmatrix} 4\varepsilon_n \\ -\varepsilon_n\end{bmatrix} ,&\rho_n\le t<\tau_{n+1}.
 \end{cases}}
 \label{eq:input}
\end{equation}

We call $[\tau_n,  \rho_n)$ the \emph{ramp} interval, and $[\rho_n, \tau_{n+1})$ the \emph{pulse} interval. Every amplitude and switching time is known before the solution is computed. 
For example, the first ramp interval is given by $[0, \frac{1}{4})$, while the first pulse interval is $[\frac{1}{4}, \frac{17}{64})$. 

\begin{theorem}[Candidate counterexample]
\label{thm:main}
The functional~\eqref{eq:V} and system~\eqref{eq:system} have the following properties.
\begin{enumerate}[label=\textup{(\roman*)}]
\item $V$ and $f$ are Lipschitz on bounded sets, $f(0,0)=0$, and
\begin{equation}
 \tfrac12|\phi(0)|^2\le V(\phi)\le\tfrac32\norm{\phi}^2.
 \label{eq:bounds}
\end{equation}
\item The system is robustly forward complete and, for every history and input value,
\begin{equation}
 \D V(\phi,f(\phi,[a \ \ b]^\top ))
 \le-p^2-\tfrac12q^2+a^2
 \le-\tfrac12|\phi(0)|^2+|[a \ \ b]^\top |^2.
 \label{eq:diss}
\end{equation}
\item Input~\eqref{eq:input} is bounded, piecewise continuous and right-continuous, with $\norm{u}_\infty=36$. From the constant history
\begin{equation}
 \phi(\theta)=[2 \ \ -4]^\top ,\qquad -1\le\theta\le0,
 \label{eq:initial}
\end{equation}
the solution satisfies, for every $n\ge0$,
\begin{equation}
 \varepsilon_np(\tau_n)>0,\qquad
 \sqrt{20+12n}\le|x(\tau_n)|\le\sqrt{20+16n}.
 \label{eq:peaks}
\end{equation}
\end{enumerate}
Consequently, the system satisfies the hypotheses of Conjecture~4 but is not ISS.
\end{theorem}

The proof will be given in the following Sections.

\section{Global Lyapunov estimates and forward completeness}
\subsection{Bounds and regularity}

For fixed \(i\in\{1,2\}\) and \(\theta\in[-1,0]\), we refer to the quantity
\[
F_i(\phi,\theta)
\]
as a \emph{branch} of the max functional
\[
V(\phi)
=
\max_{i\in\{1,2\}}
\max_{\theta\in[-1,0]}
F_i(\phi,\theta).
\]
So at $\theta=0$ we have the branches
\begin{equation}
C_\sigma(\phi):=F_\sigma(\phi,0)=g_\sigma(p)+\frac{q^2}{2}, 
 \end{equation}
from which it follows that
\begin{equation}\max_\sigma C_\sigma(\phi)=C(\phi).
\end{equation}

Hence $V\ge C\ge0$ and $E\ge0$. Dropping the two nonpositive terms in~\eqref{eq:F} proves the upper bound in~\eqref{eq:bounds}. The maximum exists by continuity in $\theta$ on a compact interval.

On a history ball of radius $R$, each of the four evaluation terms in~\eqref{eq:F} has difference at most $R\norm{\phi-\psi}$. The integral has difference at most $2R\norm{\phi-\psi}$. Uniformly in $\sigma,\theta$, this yields
\begin{equation}
 |V(\phi)-V(\psi)|\le6R\norm{\phi-\psi}
 \quad(\norm{\phi},\norm{\psi}\le R).
 \label{eq:Lip}
\end{equation}
The same property holds for $C$ and $E$. At the zero history, $V=C=E=0$, and thus $f(0,0)=0$.

\subsection{A uniform Driver estimate}
The next, sharper estimate is used both to obtain~\eqref{eq:diss} and to control the loss during a full prescribed pulse.
\begin{lemma}
\label{lem:driver}
For every history $\phi$ and input value $u=[a \ \ b]^\top$, setting $z=-q+a$ gives
\begin{equation}
 \boxed{\D V(\phi,f(\phi,u ))\le-p^2+\max\{qz,2qz\}.}
 \label{eq:sharp}
\end{equation}
\end{lemma}
\begin{proof}
Fix $\phi$ and a finite frozen velocity $w=[ r \ \ z ]^\top$, and abbreviate $\phi_h=\phi_{h,w}$. Put
\begin{equation}
 B_\sigma=(q+g_\sigma'(q))z-p^2,\qquad
 J_\sigma=g_\sigma'(p)r+qz,\qquad
 \mathcal A=\{\sigma:C_\sigma(\phi)=V(\phi)\}.
 \label{eq:BJ}
\end{equation}
No derivative of the original history is assumed. We use the elementary uniform Taylor bound
\begin{equation}
 |g_\sigma(v+d)-g_\sigma(v)-g_\sigma'(v)d|\le d^2/2,
 \label{eq:taylor}
\end{equation}
which follows by integrating the $1$-Lipschitz derivative of $g_\sigma$.

\paragraph{Old samples (i.e. $V(\phi_h)$ obtained by a branch with $\theta \in [-1,-h]$).}
For every $\xi\in[-1+h,0]$, direct substitution gives the exact identity
\begin{align}
 F_\sigma(\phi_h,\xi-h)
 ={}&F_\sigma(\phi,\xi)+\tfrac12[(q+hz)^2-q^2]
       +g_\sigma(q+hz)-g_\sigma(q)\notag\\
 &-\int_0^h(p+sr)^2\dd s
 =F_\sigma(\phi,\xi)+hB_\sigma+\mathcal R^o_\sigma(h).
 \label{eq:old}
\end{align}
For $0<h\le1$, the remainder is independent of $\xi$ and satisfies
$|\mathcal R^o_\sigma(h)|\le h^2(z^2+|pr|+r^2/3)$.
The substitution $\theta=\xi-h$ is a bijection from $[-1+h,0]$ onto $[-1,-h]$. Thus
\begin{align}
 \max_{\theta\in[-1,-h]}F_\sigma(\phi_h,\theta)
 &=\max_{\xi\in[-1+h,0]}F_\sigma(\phi,\xi)+hB_\sigma+\mathcal R^o_\sigma(h)\notag\\
 &\le V(\phi)+hB_\sigma+\mathcal R^o_\sigma(h).
 \label{eq:oldmax}
\end{align}
This is a bound on \emph{all} old samples, not a claim that any prespecified sample attains the maximum. Samples lost at the left boundary can only remove competitors.

\paragraph{New samples (i.e. $V(\phi_h)$ obtained by a branch with $\theta \in [-h,0]$).}
Every $\theta\in[-h,0]$ has the form $(\lambda-1)h$, $0\le\lambda\le1$. Exactly,
\begin{align}
 F_\sigma(\phi_h,(\lambda-1)h)
 ={}&g_\sigma(p+\lambda hr)+\tfrac12(q+hz)^2
 +g_\sigma(q+hz)-g_\sigma(q+\lambda hz)\notag\\
 &-\int_{\lambda h}^h(p+sr)^2\dd s\notag\\
 ={}&C_\sigma(\phi)+h[\lambda J_\sigma+(1-\lambda)B_\sigma]
       +\mathcal R^n_\sigma(h,\lambda).
 \label{eq:new}
\end{align}
Applying~\eqref{eq:taylor} to the three $g_\sigma$ terms and evaluating the polynomial integral proves
\begin{equation}
 |\mathcal R^n_\sigma(h,\lambda)|
 \le h^2\left(\frac56r^2+\frac32z^2+|pr|\right),
 \quad 0<h\le1,\quad0\le\lambda\le1.
 \label{eq:remainder}
\end{equation}
The old and new intervals cover the full window. Their uniform estimates therefore yield
\begin{equation}
 \D V(\phi,w)\le
 \max\bigl(\{B_{+1},B_{-1}\}\cup\{J_\sigma:\sigma\in\mathcal A\}\bigr),
 \label{eq:drivermax}
\end{equation}
with the second set omitted if empty. The constants in the remainders depend on the fixed $\phi,w$, as permitted in~\eqref{eq:driverdef}.

Now, recalling that $u=[a \ \ b]^\top, $ set $w=f(\phi,u)$. If $\sigma\in\mathcal A$, then
$C_\sigma\le C\le V=C_\sigma$, so $E=0$ and $g_\sigma'(p)=p$ (also at $p=0$). Consequently
\[
 J_\sigma=-p^2-p^6+qz\le-p^2+qz.
\]
The two values $B_\sigma$ are $-p^2+qz$ and $-p^2+2qz$, equal when $q=0$. Substitution in~\eqref{eq:drivermax} proves~\eqref{eq:sharp}.
\end{proof}

The LKF bounds and $\D V\le a^2\le|u|^2$ satisfy~\cite[Thm.~5]{survey}, which gives forward completeness.

\section{Divergence under the prescribed bounded input}

In this section we prove that the state of the system is unbounded in $\mathbb{R}_{\geq 0}$ under the provided bounded input \eqref{eq:input}, which results in system \eqref{eq:sistema} not being ISS, thus proving that Conjecture 4 is false.

With~\eqref{eq:initial} and~\eqref{eq:input}, direct substitution gives
\begin{equation}
 q(t)=\begin{cases}
 \varepsilon_n[-4+32(t-\tau_n)],&\tau_n\le t\le\rho_n,\\
 4\varepsilon_n,&\rho_n\le t\le\tau_{n+1}.
 \end{cases}
 \label{eq:q}
\end{equation}
The values agree at every boundary because $\varepsilon_{n+1}=-\varepsilon_n$. During a pulse interval the already attained equilibrium $q=4\varepsilon_n$ is maintained by $a=4\varepsilon_n$; no state reset is imposed.
The ramp input has $28\le|a|<36$, $b=0$, and the pulse input has $|u|=\sqrt{17}$. Hence the essential supremum is exactly
$36$. The input is right-continuous as well as piecewise continuous: the pulse interval lengths may tend to zero, but no finite-time accumulation occurs.

Let $\nu(t)=V(x_t)$. The constant initial history is Lipschitz, so $\nu$ is locally absolutely continuous. Direct evaluation gives $F_{+1}(\phi,\theta)=10+4\theta$, $F_{-1}(\phi,\theta)=8+4\theta$, and hence $\nu(0)=10$. Lemma~\ref{lem:driver} implies
\begin{equation}
 \dot\nu(t)\le-p(t)^2+\max\{q(t)\dot q(t),2q(t)\dot q(t)\}
 \quad\text{almost everywhere}.
 \label{eq:along}
\end{equation}
On a ramp interval, $|q|$ decreases from $4$ to $0$ and then increases to $4$. Therefore
\begin{equation}
 \int_{\tau_n}^{\rho_n}\max\{q\dot q,2q\dot q\}\dd t
 =\int_{\tau_n}^{\tau_n+1/8}q\dot q\dd t
  +\int_{\tau_n+1/8}^{\rho_n}2q\dot q\dd t=-8+16=8.
 \label{eq:budget}
\end{equation}
On a pulse interval $\dot q=0$, so $\dot\nu\le-p^2$ and $\nu$ is nonincreasing. Induction over the phases, without any growth assumption, gives
\begin{equation}
 \nu(\tau_n)\le10+8n,\qquad
 \nu(\rho_n)\le18+8n\le18(n+1).
 \label{eq:upper}
\end{equation}
Using $p^2\le2\nu$, on the entire prescribed pulse interval we obtain
\begin{equation}
 \int_{\rho_n}^{\tau_{n+1}}p(t)^2\dd t
 \le36(n+1)\delta_n=\frac9{16}.
 \label{eq:pulsecost}
\end{equation}
Crucially, this estimate does not stop at a level crossing and does not presume that any level is reached.

We prove by induction that
\begin{equation}
 \varepsilon_np(\tau_n)>0,\qquad
 H_n:=\tfrac12p(\tau_n)^2\ge2+6n.
 \label{eq:induction}
\end{equation}
The claim holds initially. Suppose it holds at $n$. On the ramp interval $b=0$, so $\dot p=-p-p^5$, its nonzero sign is preserved, and
\begin{equation}
 \frac{\mathrm d}{\mathrm dt}\arctan(p^2)=-2p^2,\qquad
 I_n:=\int_{\tau_n}^{\rho_n}p(t)^2\dd t\le\frac\pi4<1.
 \label{eq:dampcost}
\end{equation}
Also $\frac{d}{dt}(p^{-4})=4(1+p^{-4})>4$. Integration over the ramp interval of length $1/4$ gives $|p(\rho_n)|<1$.

Track the sample at $\tau_n$ with sign index $\sigma=\varepsilon_n$. Its value at the end of the ramp interval is
\begin{equation}
 R_n:=F_{\varepsilon_n}(x_{\rho_n},\tau_n-\rho_n)
     =H_n+16-I_n.
 \label{eq:R}
\end{equation}
Indeed, the recorded $p$ has sign $\varepsilon_n$, the recorded $q=-4\varepsilon_n$ gives $g_{\varepsilon_n}(q)=0$, and the final $q=4\varepsilon_n$ contributes $q^2/2+g_{\varepsilon_n}(q)=16$. Thus
\begin{equation}
 R_n>17+6n\ge16,\qquad
 6(n+1)\le R_n\le V(x_{\rho_n})\le18(n+1).
 \label{eq:Rbounds}
\end{equation}
The tracked sample need not realize the maximum. Only its contribution as a lower bound for $V$ is used.

Write $R=R_n$ and $k=n+1$. The same sample remains in the delay window throughout the pulse interval, since
\[
 t-\tau_n\le\frac14+\delta_n\le\frac{17}{64}<1.
\]
With $q$ constant, its score is $S(t)=R-\int_{\rho_n}^tp(s)^2\dd s$. By~\eqref{eq:pulsecost},
\begin{equation}
 V(x_t)\ge S(t)\ge R-\frac9{16}
 \qquad(\rho_n\le t\le\tau_{n+1}).
 \label{eq:retained}
\end{equation}
Set
\begin{equation}
 y=-\varepsilon_n p,\qquad L_n=\sqrt{2R-18}>1.
 \label{eq:level}
\end{equation}
The ramp estimate gives $-1<y(\rho_n)<0$. Under the prescribed input $b=-\varepsilon_n$,
\begin{equation}
 \dot y=-y-y^5+16(1+V)^3E.
 \label{eq:y}
\end{equation}
At $y=-1$ its right-hand side is $2+16(1+V)^3E\ge2$, so the solution cannot cross below $-1$ during the pulse interval.

Whenever $-1\le y\le L_n$, we have $C=y^2/2+8\le R-1$. Together with~\eqref{eq:retained}, this gives
\begin{equation}
 E=V-C\ge\frac7{16},\qquad 1+V\ge R.
 \label{eq:gap}
\end{equation}
Since $R\ge16$ and $|y|\le\sqrt{2R}$,
\[
 y+y^5\le\sqrt{2R}(1+4R^2)\le8R^{5/2}\le2R^3.
\]
Consequently,
\begin{equation}
 \dot y\ge7R^3-2R^3\ge4R^3
 \quad\text{whenever }-1\le y\le L_n.
 \label{eq:speed}
\end{equation}
If $y$ did not reach $L_n$ before the pulse interval ended, integrating this inequality and using~\eqref{eq:Rbounds} would give
\begin{equation}
 y(\tau_{n+1})\ge-1+4R^3\delta_n
 =-1+\frac{R^3}{16k}
 \ge-1+\frac{27}{2}k^2
 >6\sqrt{k}\ge L_n,
 \label{eq:crossing}
\end{equation}
which is a contradiction. Here the strict inequality holds for every integer $k\ge1$, and $L_n\le\sqrt{2R}\le6\sqrt{k}$.

There is no input switch at this crossing: the pulse continues to its predetermined endpoint. At every interior time with $y=L_n$, estimate~\eqref{eq:speed} still holds. The field along the solution is continuous inside the constant-input phase. A first downward crossing would therefore have a nonpositive crossing derivative, contradicting~\eqref{eq:speed}. Thus, by continuity also at the endpoint,
\begin{equation}
 y(\tau_{n+1})\ge L_n>0.
 \label{eq:noreturn}
\end{equation}

Equation~\eqref{eq:noreturn} gives the sign needed at cycle $n+1$, since $\varepsilon_{n+1}=-\varepsilon_n$, and
\begin{equation}
 H_{n+1}\ge L_n^2/2=R_n-9=H_n+7-I_n>H_n+6.
 \label{eq:increment}
\end{equation}
This closes~\eqref{eq:induction}. Because $q(\tau_n)^2=16$, its lower estimate gives $|x(\tau_n)|^2\ge20+12n$. The upper estimate follows from $|x|^2\le2V$ and~\eqref{eq:upper}. Therefore~\eqref{eq:peaks} holds, while $\tau_n\ge n/4\to\infty$.

An ISS estimate~\eqref{eq:iss} for this fixed history and bounded input would imply
\[
 |x(t)|\le\beta(\sqrt{20},t)+\mu(36)
 \le\beta(\sqrt{20},0)+\mu(36)<\infty\qquad(t\ge0),
\]
contradicting~\eqref{eq:peaks}. All the conjecture's functions are explicit:
\begin{equation}
 \underline\alpha(s)=\tfrac12s^2,\qquad
 \overline\alpha(s)=\tfrac32s^2,\qquad
 \alpha(s)=\tfrac12s^2,\qquad\gamma(s)=s^2.
 \label{eq:classes}
\end{equation}

\section{Numerical simulation of the state evolution for the given input}

Fig.~\ref{fig:state} shows current-state components for system \eqref{eq:sistema} under the prescribed input \eqref{eq:input}, exhibiting unbounded growth and contradicting ISS. This provides a numerical validation of Theorem \ref{thm:main}.

\begin{figure}[H]
\centering\includegraphics[width=\linewidth]{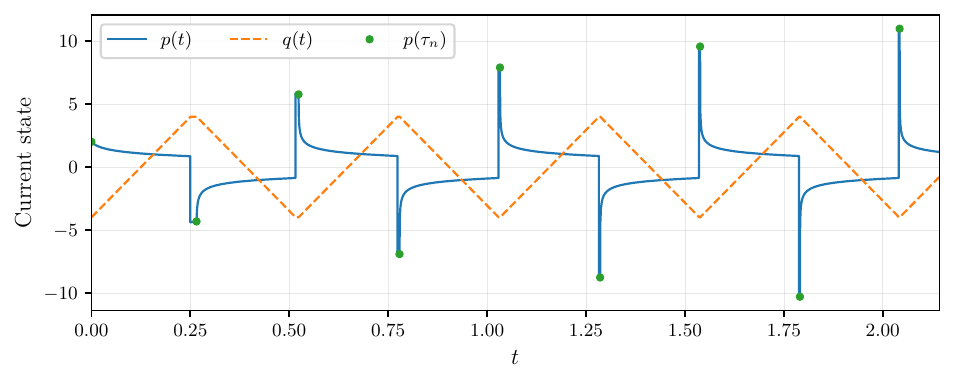}
\caption{Computed current state $x(t)=[ p(t) \ \ q(t)]^\top$ for  input \eqref{eq:input}. The state is continuous, even though steep pulse fronts appear almost vertical on this global time scale.}
\label{fig:state}
\end{figure}

\section{Conclusions}
We have exhibited a two-dimensional retarded functional differential equation and an LKF with pointwise dissipation rate which answer Conjecture~4 of \cite[Section~8.3]{survey} in a negative sense.

By using a piecewise continuous and right continuous input signal, the solution started from a suitable constant initial state is unbounded, thus violating ISS.  

\section*{Acknowledgments}
This work was supported in part by a grant of access to OpenAI models through the ChatGPT for Academic Researchers program. Access to Claude was provided by Anthropic through the Claude Team plan for scientists.


\begin{thebibliography}{9}
\bibitem{survey}
A.~Chaillet, I.~Karafyllis, P.~Pepe, and Y.~Wang,
``The ISS framework for time-delay systems: a survey,''
\emph{Mathematics of Control, Signals, and Systems}, vol.~35, pp.~237--306, 2023.
\href{https://doi.org/10.1007/s00498-023-00341-w}{doi:10.1007/s00498-023-00341-w}.
\end{thebibliography}
\end{document}